\documentclass[journal]{IEEEtran}
\usepackage{amsthm}
\usepackage{amssymb}
\usepackage{graphicx}
\usepackage{amsmath}
\usepackage{newtxmath}
\usepackage[cal=cm]{mathalpha} 
\usepackage{tablefootnote} 
\usepackage{threeparttable} 
\usepackage{bm}
{}
{}
{}
\newtheorem{proposition}{Proposition}{}
{}
{}

\usepackage{algorithm}
\usepackage{algpseudocode}
\usepackage{graphics}
\usepackage{epsfig}
\usepackage{multirow}
\usepackage{caption}
\usepackage{array}
\usepackage{cite}
\usepackage{stfloats}
\usepackage{subfig} 
\usepackage{booktabs}
\usepackage{comment}
\usepackage{color}
\usepackage{tabu} 
\usepackage{stfloats}
\usepackage{mathtools}
\usepackage{stmaryrd} 
\usepackage[
colorlinks=true,
linkcolor=blue,
anchorcolor=blue,
citecolor=blue,
bookmarks=true
]{hyperref}
\usepackage{etoolbox}
\makeatletter

\let\OriginalCite\cite

\newcommand{\MarkFirstCitation}[1]{%
	\ifcsundef{FirstCite@#1}{%
		\csgdef{FirstCite@#1}{1}%
		\Hy@raisedlink{%
			\hyper@anchorstart{FirstCite:#1}%
			\hyper@anchorend
		}%
	}{}%
}
\renewcommand{\cite}[2][]{%
	\ifstrempty{#1}{%
		\OriginalCite{#2}%
	}{%
		\OriginalCite[#1]{#2}%
	}%
	\forcsvlist{\MarkFirstCitation}{#2}%
}
\let\OldBibLabel\@biblabel
\def\CurrentBibKey{}

\renewcommand{\@biblabel}[1]{%
	\ifx\CurrentBibKey\@empty
	\OldBibLabel{#1}%
	\else
	\ifcsname FirstCite@\CurrentBibKey\endcsname
	\begingroup
	\hypersetup{linkcolor=black}%
	\hyperlink{FirstCite:\CurrentBibKey}{%
		\OldBibLabel{#1}%
	}%
	\endgroup
	\else
	\OldBibLabel{#1}%
	\fi
	\fi
}
\let\OldAtBibitem\@bibitem
\def\@bibitem#1{%
	\gdef\CurrentBibKey{#1}%
	\OldAtBibitem{#1}%
}

\let\OldAtLbibitem\@lbibitem
\def\@lbibitem[#1]#2{%
	\gdef\CurrentBibKey{#2}%
	\OldAtLbibitem[#1]{#2}%
}

\makeatother
\usepackage[capitalize]{cleveref} 
\crefname{figure}{Fig.}{Figs.}   
\Crefname{figure}{Fig.}{Figs.}   
\usepackage{subcaption}
\makeatletter
\let\sum\relax
\let\prod\relax
\DeclareSymbolFont{largesymbols}{OMX}{cmex}{m}{n}
\DeclareMathSymbol{\sum}{\mathop}{largesymbols}{"50}
\DeclareMathSymbol{\prod}{\mathop}{largesymbols}{"51}
\DeclareMathOperator*{\argmax}{argmax}

\DeclareMathOperator*{\diag}{\text{diag}}
\DeclareMathOperator*{\rank}{\text{rank}}

\begin{document}
	\title{{\huge Self-Calibration DOA Estimation for Movable Antenna Systems with Antenna Position Errors}}
	
	
	\author{\IEEEauthorblockN{
			Chengzhi Ye, \IEEEmembership{Student Member,~IEEE}, 
			Ruoyu Zhang, \IEEEmembership{Senior Member,~IEEE}, 
			Wen Wu, 
			\IEEEmembership{Senior Member,~IEEE},
			Byonghyo Shim,
			\IEEEmembership{Fellow,~IEEE}
		}
		\thanks{
			
			Chengzhi Ye, Ruoyu Zhang, and Wen Wu are with the Key Laboratory of Near-Range RF Sensing ICs and Microsystems (NJUST), Ministry of Education, School of Electronic and Optical Engineering, Nanjing University of Science and Technology, Nanjing 210094, China. Chengzhi Ye is also with Qian Xuesen College, Nanjing University of Science and Technology, Nanjing 210094, China (e-mail: qxsycz8166@njust.edu.cn; ryzhang19@njust.edu.cn; 
			wuwen@njust.edu.cn).
			Byonghyo Shim is with the Department of Electrical and Computer Engineering and the Institute of New Media and Communications, Seoul National University, Seoul 08826, South Korea (e-mail:
			bshim@snu.ac.kr). \textit{(Corresponding author: Ruoyu Zhang)}}
		\vspace{-2em}
	}

	
	\maketitle
	
	

	%
	\IEEEpeerreviewmaketitle

	\begin{abstract}
		In this letter, we investigate the direction-of-arrival (DOA) estimation problem for wireless sensing with movable antenna (MA) systems in the presence of unknown antenna position errors (APE). To achieve robust wireless sensing, we transform the DOA estimation problem with APE into an optimization problem via the orthogonality between the steering vector and the noise subspace. Then we propose an alternating optimization (AO)-based self-calibration estimation, which consists of two stages and iteratively estimates the APE and DOA. Specifically, in the first stage, by fixing the APE, the problem reduces to the classical DOA estimation problem, which is solved using the multiple signal classification algorithm. In the second stage, we fix the DOA to estimate the APE. By applying the Lagrange multiplier technique to the subproblem, we obtain a closed-form expression for the APE estimation.  Simulation results demonstrate the superior DOA estimation performance of the proposed self-calibration algorithm for MA systems compared to the existing approaches.
	\end{abstract}
	\section{Introduction}
	
	Direction-of-arrival (DOA) estimation is a fundamental problem in array signal processing, with extensive applications in radar sensing and wireless communications \cite{DOAHAD2022Ruoyu,ye2026RAsensing,MovingSampling2DDOA}. To achieve high-resolution parameter estimation, conventional systems typically employ fixed-position arrays (FPAs), such as uniform linear arrays and various sparse arrays \cite{ye2026tensordecompositionbasedwirelesssensing,ISACRZzhang2024}. 
	However, these traditional systems typically rely on antennas arranged with fixed spacing, which restricts the flexible acquisition of signal information in the continuous spatial domain.

	Recently, movable antenna (MA) arrays have garnered significant research interest in the fields of communications and sensing, owing to their ability to reshape channel conditions through intelligent antenna positioning \cite{zhuMovableAntennasWireless2024, ChengzhiIOTJ20262D}. 
	 In contrast to conventional FPAs, MA systems allow for continuous position adjustment, thereby unlocking the potential to fully exploit the spatial degrees of freedom within a confined region \cite{C11456856,huTwotimescaleDesignMovable2025}. 
	Various studies have demonstrated the superior wireless sensing capabilities of MA systems, showing that optimized antenna placement can significantly enhance estimation accuracy. The authors in \cite{maMovableAntennaEnhanced2024} proposed an MA-enabled wireless sensing system and optimized the antenna positions to minimize the Cram\'{e}r-Rao bound (CRB) of DOA estimation, which implies that MAs can potentially outperform conventional FPA in terms of estimation accuracy and ambiguity reduction. A tensor decomposition-based method has been proposed in \cite{zhangChannelEstimationMovableantenna2024} to efficiently reconstruct the wireless channel between arbitrary Tx/Rx MA positions by leveraging a two-stage successive antenna movement pattern.
In the realm of wireless communications, MAs have also been extensively explored to enhance the channel capacity and reliability by fully exploiting spatial diversity \cite{NearMultiuser2025,S11363456, ZhengRotatableComm2026, Zheng2025RotatableOpportunities}.
	Despite the potential mentioned above in both wireless sensing and communications, existing studies typically rely on idealized antenna array models, assuming that the received signals are free of practical array imperfections.

	Over the past decades, extensive research efforts have been made to mitigate array imperfections in conventional FPA, such as mutual coupling and gain/phase errors \cite{MutualCoupling2016,YaoleiTensor2025,liuEigenstructureMethodEstimating2011,DOAGPEChengzhi2025}. Apart from these impairments, antenna position errors (APE) can also significantly degrade the DOA performance \cite{FLANAGAN20012201,SongPositionSmall2024,ArrayPosition2025}. To address this issue, a joint estimation framework based on a
	generalized incoherently distributed source model has been proposed
	in \cite{ArrayPosition2025} simultaneously recovering DOA, which is applicable
	under small APE and FPAs. In MA systems, antennas often undergo rapid movement to adapt to dynamic channel environments, inevitably resulting in significant APE, which renders current methods inapplicable.

	To tackle the challenge of DOA estimation  with APE in MA systems, we propose a novel alternative optimization (AO)-based self-calibration algorithm estimating both DOA and APE. To be specific, we calculate the covariance matrix of the received signals and perform eigenvalue decomposition to obtain the signal subspace and the noise subspace. Based on the orthogonality of the steering vectors and the noise subspace, we formulate the joint estimation problem of DOA and APE, and develop an AO-based self-calibration algorithm to estimate
	the DOA and APE in two stages. In the first stage, we fix the APE such that the problem reduces to the conventional DOA estimation problem, which can be solved by the multiple signal classification (MUSIC) algorithm. In the second stage, we fix the DOA to estimate the APE such that the problem is transformed into a linearly constrained quadratic minimization problem, which can be
	subsequently solved through the Lagrange multiplier method
	to estimate the APE. Furthermore, we provide a thorough analysis of the computational complexity for the proposed AO-based self-calibration algorithm. Comprehensive numerical simulations demonstrate the superior DOA estimation accuracy and robustness of the proposed approach compared to conventional MUSIC methods using either uncalibrated full or limited calibrated antennas.
	
	Notations: Symbols for vectors (lower case) and matrices (upper case) are in boldface. 
	The symbols $\circledast$ and $\dagger $ denote the Hadamard product and pseudoinverse, respectively. The $T$-dimensional vector with all elements equal to 1 is written as $\bm{1}_{T \times 1}$. The expression $[\bm{\alpha}]_k$ denotes the $k$-th element of $\bm{\alpha}$. The expression $[\bm{B}]_{m,:}$ and $[\bm{B}]_{:,k}$
	denote the $m$-th row and the $k$-th column of
	$\bm{B}$, respectively.
	\section{System Model and Problem Formulation}
	
	\begin{figure}[t]
		\vspace{-0.5em}
		\centering
		\includegraphics[width=0.85\linewidth]{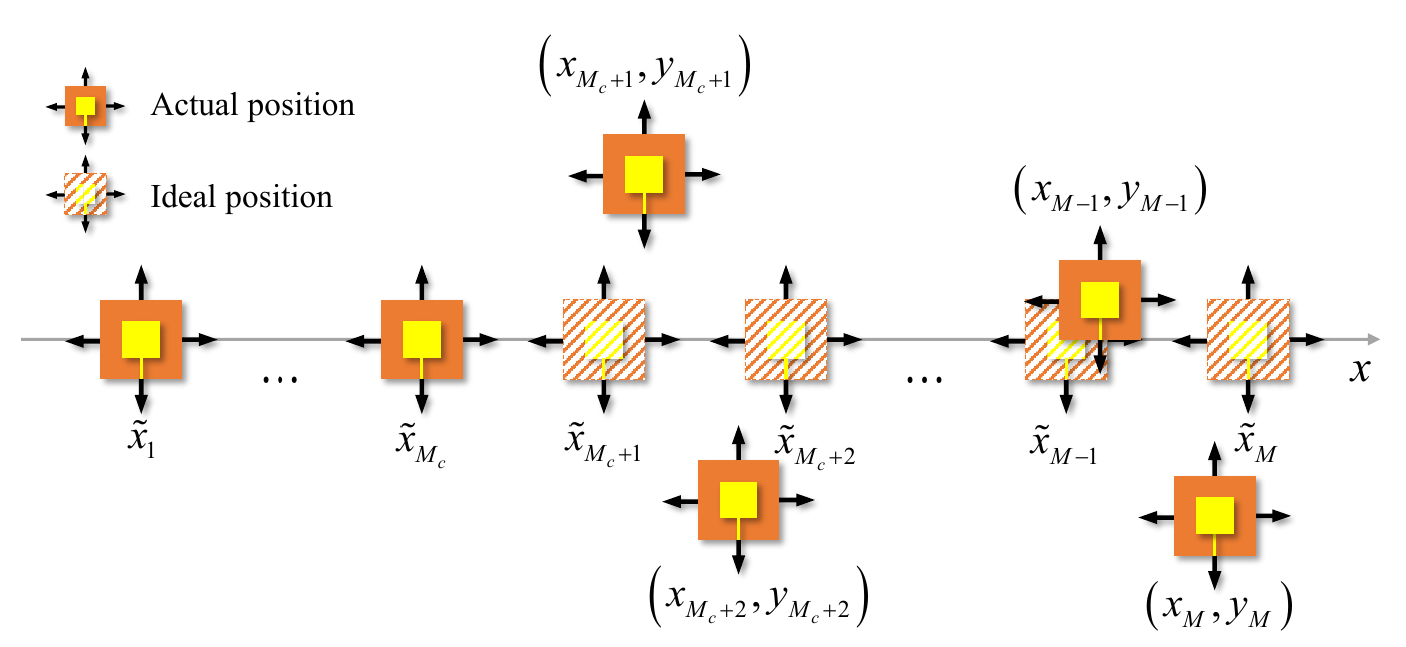}
		\caption{The diagram of the APE model in MA systems.}
		\label{PE}
		\vspace{-0.75em}
	\end{figure}
	
	In this section, we consider a DOA estimation model with APE caused by the moving of MA. We assume $K$ uncorrelated narrowband source signals impinging on an array equipped with $M$ MAs, whose normal positions are expressed as $\bm{\tilde{P}} =[\bm{\tilde{x}}, \bm{0}_{M\times 1}]$ with $\bm{\tilde{x}}=[\tilde{x}_1,\tilde{x}_2,\ldots,\tilde{x}_M]^{\mathrm{T}}$. Here, $\tilde{x}_m \in [0, H]$ denotes the ideal position of the $m$-th MA within a movable region of length $H$. 
	As shown in \Cref{PE}, the APE of MAs can be denoted by
	\begin{align}
		{\Delta\bm{P}}=[ \Delta \bm{x},\Delta \bm{y}]\in \mathbb{C}^{M\times 2},
	\end{align}
	where $\Delta \bm{x}=[\Delta x_1,\Delta x_2,\ldots,\Delta x_M]^{\mathrm{T}}$ and $\Delta \bm{y}=[\Delta y_1,\Delta y_2,\ldots,\Delta y_M]^{\mathrm{T}}$ represent the position error along the $x$ and $y$-axis, respectively. Without loss of generality, the first $M_c$ antennas are calibrated, i.e. $\Delta x_m=0$ for $m=1,\ldots,M_c$ and $\Delta y_m=0$ for $m=1,\ldots,M_c$. 
	Thus, the actual MA position is $\bm{{P}}=[\bm{{x}},\bm{{y}}]\in \mathbb{C}^{M\times 2}$, where $\bm{{x}}=[{x}_1,{x}_2,\ldots,{x}_M]^{\mathrm{T}}$, $\bm{{y}}=\Delta \bm{y}$, and ${x}_m=\tilde{x}_m+\Delta x_m$ for $m=1,\ldots,M$. The $t$-th snapshot of the received signal $\bm{z}$ can be given by
	\begin{align}\label{ysignal}
		\bm{z}(t)=\sum\nolimits_{k=1}^{K}{\bm{a}(\theta_k,\bm{{P}})s_k(t)+\bm{n}(t)},
	\end{align}
	where \begin{align}
		\bm{a}(\theta_k,\bm{{P}})=&\big[e^{j\kappa({x}_1\cos \theta_k+{y}_1\sin \theta_k)},e^{j\kappa({x}_2\cos \theta_k+{y}_2\sin\theta_k)}\nonumber\\ &,\ldots,e^{j\kappa({x}_M\cos \theta_k+{y}_M\sin\theta_k)}\big]^{\mathrm{T}}
	\end{align}
	 is the actual steering vector. $\kappa$ is the wavenumber. $\theta_k \in [\vartheta_{\min}, \vartheta_{\max}]$ is the DOA of the $k$-th target. $s_k(t)$ is the $k$-th signal at $t$-th snapshot. $\bm{n}(t)\sim\mathcal{CN}(\bm{0}_{M\times 1},\sigma_n^2\bm{I}_{M\times M})$ is the additive white Gaussian noise, and $\sigma_n^2$ is the noise power. Considering $T$ snapshots, \eqref{ysignal} can be rewritten as
	\begin{align}\label{Ysignal}
		\bm{Z}=\bm{{A}}(\bm{\theta},\bm{{P}})\bm{S}+\bm{N},
	\end{align}
	where $\bm{Z}=[\bm{z}(1),\bm{z}(2),\ldots,\bm{z}(T)]\in \mathbb{C}^{M\times T}$, $\bm{{A}}(\bm{\theta},\bm{{P}})=[\bm{a}(\theta_1,\bm{{P}}),\bm{a}(\theta_2,\bm{{P}}),\ldots,\bm{a}(\theta_K,\bm{{P}})]\in \mathbb{C}^{M\times K}$, and $\bm{S}=[\bm{s}(1),\bm{s}(2),\ldots,\bm{s}(T)]\in \mathbb{C}^{K\times T}$ with $\bm{s}(t)=[s_1(t),s_2(t),\ldots,s_K(t)]^{\mathrm{T}}\in \mathbb{C}^{K\times 1}$.
	Subsequently, we can obtain the covariance matrix ${\bm{R}_Z}$ of the received $\bm{Z}$ from \eqref{Ysignal}, i.e.
	\begin{align}\label{RY}
		\bm{{R}}_Z=\mathbb{E}\{\bm{Z}\bm{Z}^{\mathrm{H}}\}=\bm{{A}}(\bm{\theta},\bm{{P}})\bm{R}_S\bm{{A}}(\bm{\theta},\bm{{P}})^{\mathrm{H}}+\sigma_n^2\bm{I}_{M\times M},
	\end{align}
	where $\bm{R}_S=\mathbb{E}\{\bm{S}\bm{S}^{\mathrm{H}}\}=\diag(\sigma_1^2,\sigma_2^2,\ldots,\sigma_K^2)$ represents the transmitted signal covariance matrix with $\sigma_k^2$ denoting the $k$-th power of the signals. By applying the eigenvalue decomposition, we have
	\begin{align}\label{EVD}
		\bm{{R}}_Z=\sum\nolimits_{m=1}^M\xi_m\bm{e}_m\bm{e}_m^{\mathrm{H}}=\bm{E}_S\bm{\varXi}_S\bm{E}_S^{\mathrm{H}}+\bm{E}_N\bm{\varXi}_N\bm{E}_N^{\mathrm{H}},
	\end{align}
	where $\bm{\varXi}=\diag(\xi_1,\xi_2,\ldots,\xi_M) \in \mathbb{C}^{M\times M}$, with $\xi_m$ denoting the $m$-th largest eigenvalue and $\bm{e}_m$ denoting the corresponding eigenvector. $\bm{E}_S=[\bm{e}_1,\bm{e}_2,\ldots,\bm{e}_K] \in \mathbb{C}^{M\times K}$ and  $\bm{E}_N=[\bm{e}_{K+1},\bm{e}_{K+2},\ldots,\bm{e}_M]\in \mathbb{C}^{M\times (M-K)}$ represent the signal subspace and noise subspace, respectively. Since the signal subspace is orthogonal to the noise subspace, we have
	\begin{align}\label{org}
		\|\bm{E}_N^{\mathrm{H}}{\bm{a}}(\theta_k,\bm{{P}})\|_F^2=0,\enspace k=1,2,\ldots,K.
	\end{align}
	By exploiting the elementary identity $\|\bm{b}\|_F^2=\bm{b}^{\mathrm{H}}\bm{b}$, the optimization problem can be formulated as
	\begin{subequations}\label{P1}
		\begin{align}
			\mathrm{(P1)}\enspace \min_{\theta_k,\bm{{P}}}\enspace&\bm{a}(\theta_k,\bm{{P}})^{\mathrm{H}}\bm{E}_N\bm{E}_N^{\mathrm{H}} \bm{a}(\theta_k,\bm{{P}})\\
			\mathrm{s.t.} \enspace &\Delta x_m = 0, \enspace m = 1, \ldots, M_c\\
			&\Delta y_m = 0, \enspace m = 1, \ldots, M_c.
		\end{align}
	\end{subequations}

	Our goal is to estimate both the DOA and the APE of the MA systems for the subsequent mechanical calibration.
	However, the mobility of MAs introduces uncertainty regarding their positions, causing a DOA estimation ambiguity. To address this issue, we will propose an AO-based self-calibration algorithm.
	
	\section{Proposed AO-Based Self-Calibration Algorithm}
	
	In this section, we propose an AO-based self-calibration algorithm to jointly estimate the DOA and the APE. In the first stage, we fix the APE and then estimate the DOA via the MUSIC algorithm. In the second stage, we fix the DOA and calculate the APE by applying the Lagrange multiplier method. These two stages are executed iteratively until the predefined convergence criterion is met. We also analyze the computational complexity of the proposed algorithm.
	
\subsection{Stage 1: DOA Estimation}
	In this subsection,  we focus on estimating the DOA while fixing the APE. The minimization problem (\hyperref[P1]{P1}) can reduce to
	\begin{align}\label{P2}
		\mathrm{(P2)}\enspace\min_{\theta_k}\enspace \bm{a}(\theta_k,\bm{\hat{{P}}})^{\mathrm{H}}\bm{E}_N\bm{E}_N^{\mathrm{H}} \bm{a}(\theta_k,\bm{\hat{{P}}}),
	\end{align}
	where $\bm{\hat{{P}}} = \bm{\tilde{P}}+\Delta \hat{\bm{P}}$ represents the estimated actual MA position with $\Delta \hat{\bm{P}}$ denoting the estimated APE. We can transform the continuous minimization problem into a one-dimensional exhaustive search via MUSIC algorithm, i.e.,
	\begin{align}\label{Initial}
		\hat{\theta}_k = \argmax_{\vartheta\in [\vartheta_{\min}, \vartheta_{\max}]} {1}/{\big[\bm{a}(\vartheta,\bm{\hat{{P}}})^{\mathrm{H}}\bm{E}_N\bm{E}_N^{\mathrm{H}} \bm{a}(\vartheta,\bm{\hat{{P}}})\big]},
	\end{align}
	where $\hat{\theta}_k$ denotes the estimated DOA of the $k$-th target.
	
	\subsection{Stage 2: APE Estimation}
	In this subsection, with a fixed DOA estimation, we focus on estimating the APE. By substituting the estimated DOA, the objective function can be simplified as \eqref{min}. The first equation holds true based on the relationship in \eqref{Hadamard}. The third equation holds true by employing the identity $\bm{b}\circledast \bm{c}=\diag(\bm{c})\bm{b}$. The fourth equation holds true by defining $\bm{Q}(\hat{\theta}_k,\bm{\tilde{P}})=\big\{\diag\big[\bm{a}(\hat{\theta}_k,\bm{{P}})\big]\big\}^{\mathrm{H}} \bm{E}_N\bm{E}_N^{\mathrm{H}} \diag\big[\bm{a}(\hat{\theta}_k,\bm{{P}})\big] \succeq \bm{0}_{M \times M}\in \mathbb{C}^{M\times M}$, respectively. 
	\begin{figure*}
		\vspace{-0.75em}
		\begin{align}\label{min}
			&\bm{a}(\hat{\theta}_k,\bm{{P}})^{\mathrm{H}}\bm{E}_N\bm{E}_N^{\mathrm{H}} \bm{a}(\hat{\theta}_k,\bm{{P}})
			=\big[\bm{{a}}(\hat{\theta}_k,\bm{\tilde{P}})\circledast\bm{a}(\hat{\theta}_k,{\Delta\bm{P}})\big]^{\mathrm{H}} \bm{E}_N\bm{E}_N^{\mathrm{H}} \big[\bm{a}(\hat{\theta}_k,\bm{\tilde{P}})\circledast\bm{a}(\hat{\theta}_k,{\Delta\bm{P}})\big]\nonumber\\
			=&\big(\diag\big[\bm{a}(\hat{\theta}_k,\bm{\tilde{P}})\big]\bm{a}(\hat{\theta}_k,{\Delta\bm{P}})\big)^{\mathrm{H}} \bm{E}_N\bm{E}_N^{\mathrm{H}} \big(\diag\big[\bm{a}(\hat{\theta}_k,\bm{\tilde{P}})\big]\bm{a}(\hat{\theta}_k,{\Delta\bm{P}})\big)\nonumber\\
			=&\bm{a}(\hat{\theta}_k,{\Delta\bm{P}})^{\mathrm{H}}\big(\diag\big[\bm{a}(\hat{\theta}_k,\bm{\tilde{P}})\big]\big)^{\mathrm{H}} \bm{E}_N\bm{E}_N^{\mathrm{H}} \diag\big[\bm{a}(\hat{\theta}_k,\bm{\tilde{P}})\big]\bm{a}(\hat{\theta}_k,{\Delta\bm{P}})
			=\bm{a}(\hat{\theta}_k,{\Delta\bm{P}})^{\mathrm{H}}\bm{Q}(\hat{\theta}_k,\bm{\tilde{P}})\bm{a}(\hat{\theta}_k,{\Delta\bm{P}}).
		\end{align}
		\hrule
		\vspace{-0.75em}
	\end{figure*}
	\begin{figure*}
		\vspace{-0.75em}
		\begin{align}\label{Hadamard}
			\!\!\!\bm{a}(\hat{\theta}_k,\!\bm{{P}})\!=&\big[e^{j\kappa[(\tilde{x}_1+\Delta x_1)\cos \hat{\theta}_k+\Delta{y}_1\sin \hat{\theta}_k]},\ldots,e^{j\kappa[(\tilde{x}_M+\Delta x_M)\cos \hat{\theta}_k+\Delta y_M\sin\hat{\theta}_k]}\big]^{\mathrm{T}}\\
			\!=&\big[e^{j\kappa\tilde{x}_1\cos\theta_k},\ldots,e^{j\kappa\tilde{x}_M\cos\theta_k}\big]^{\mathrm{T}}\!\!\circledast\!\big[e^{j\kappa(\Delta x_1\cos\theta_k+\Delta y_1\sin\hat{\theta}_k)},\ldots,e^{j\kappa(\Delta x_M\cos\hat{\theta}_k+\Delta y_M\sin\hat{\theta}_k)}\big]^{\mathrm{T}}\!=\bm{{a}}(\theta_k,\!\bm{\tilde{P}})\!\circledast\!\bm{a}(\hat{\theta}_k,\!{\Delta\bm{P}}).\nonumber
		\end{align}
			%
		\hrule
		\vspace{-0.75em}
	\end{figure*}

	Next, we study the constraint condition as follows. It can be noted that the steering vector of APE in \eqref{Hadamard} satisfies
	\begin{align}\label{Lenear}
		[\bm{a}(\hat{\theta}_k,{\Delta\bm{P}})]_m=1, \enspace m=1,\ldots,M_c.
	\end{align}
	To facilitate computation, we introduce the coefficient matrix $\bm{W}\in \mathbb{C}^{M\times M_c}$, which can be given by
	\begin{align}
		\bm{W}=\begin{bmatrix}
			\bm{I}_{M_c\times M_c}\\
			\bm{0}_{(M-M_c)\times M_c}
		\end{bmatrix}.
	\end{align}
	Thus, the constraint condition in (\hyperref[P1]{P1}) can be rewritten in a matrix form, i.e.,
	\begin{align}\label{condition}
		\bm{W}^{\mathrm{H}}\bm{a}(\hat{\theta}_k,{\Delta\bm{P}})=\bm{1}_{M_c\times 1}
	\end{align}
	Using \eqref{condition}, the minimization problem \eqref{min} can be rewritten as
\begin{subequations}\label{P3}
		\begin{align}
		\mathrm{(P3)}\enspace\min_{{\Delta\bm{P}}} \enspace& \bm{a}(\hat{\theta}_k,{\Delta\bm{P}})^{\mathrm{H}}\bm{Q}(\hat{\theta}_k,\bm{\tilde{P}})\bm{a}(\hat{\theta}_k,{\Delta\bm{P}}) \\ 
		\text{s.t.}\enspace& 	\bm{W}^{\mathrm{H}}\bm{a}(\hat{\theta}_k,{\Delta\bm{P}})=\bm{1}_{M_c\times 1}
	\end{align}
\end{subequations}
	It should be noted that the minimization problem (\hyperref[P3]{P3}) is non-convex with respect to $\Delta\bm{P}$ but translates into a convex optimization problem with respect to $\bm{a}(\hat{\theta}_k,{\Delta\bm{P}})$. This convexity enables the application of the Lagrange multiplier method. However, the existence of a direct and unique closed-form solution inherently depends on the rank of $\bm{Q}(\hat{\theta}_k,\bm{\tilde{P}})$. We subsequently introduce Proposition \ref{Propose1} to analyze the rank of $\bm{Q}(\hat{\theta}_k,\bm{\tilde{P}})$.
	\begin{proposition}\label{Propose1}
		For $k=1,2,\ldots,K$, $\bm{Q}(\hat{\theta}_k ,\bm{\tilde{P}})$ is a singular matrix with a rank of $M-K$.
	\end{proposition}
	\begin{proof}
		Based on \eqref{min}, we can deduce the rank of $\bm{Q}(\hat{\theta}_k ,\bm{P})$ as
		\begin{align}
			&\rank\big[\bm{Q}(\hat{\theta}_k ,\bm{\tilde{P}})\big]\nonumber\\
			=&\rank \Big[\big\{\diag\big[\bm{a}(\hat{\theta}_k ,\bm{\tilde{P}})\big]\big\}^{\mathrm{H}} \bm{E}_N\bm{E}_N^{\mathrm{H}} \diag\big[\bm{a}(\hat{\theta}_k ,\bm{\tilde{P}})\big]\Big]\nonumber\\
			=&\rank\Big[\big\{\diag\big[\bm{a}(\hat{\theta}_k ,\bm{\tilde{P}})\big]\big\}^{\mathrm{H}} \bm{E}_N\Big]
			=M-K,
		\end{align}
		where the second equality follows from the fundamental identity $\rank(\bm{XX}^{\mathrm{H}})=\rank(\bm X)$, and the final equality holds strictly because multiplying the matrix $\bm{E}_N$ by a non-singular diagonal matrix $\big\{\diag\big[\bm{a}(\hat{\theta}_k ,\bm{\tilde{P}})\big]\big\}^{\mathrm{H}}$ preserves its rank with $\rank(\bm{E}_N)=M-K$.
	\end{proof}
	Proposition \ref{Propose1} demonstrates that $\bm{Q}(\hat{\theta}_k ,\bm{\tilde{P}})$ experiences a rank deficiency of $K$. This deficiency essentially hinders the direct computation of the inverse matrix typically required in the standard Lagrange multiplier method. To resolve this invertibility issue, we introduce a slight diagonal perturbation
	\begin{align}
		\bm{\bar{Q}}(\hat{\theta}_k ,\bm{\tilde{P}})=\bm{Q}(\hat{\theta}_k ,\bm{\tilde{P}})+\varepsilon\bm{I}_{M \times M},
	\end{align}
	where $\varepsilon$ is a sufficiently small positive constant. This perturbation maintains the inherent matrix structure while simultaneously guaranteeing stable numerical matrix inversion.
	
	Let $\mathcal{L} \big[\bm{a}(\hat{\theta}_k,{\Delta\bm{P}}),\bm{\mu}\big]$ be the Lagrange function, which can be formulated as
	\begin{align}\label{Lag}
		\mathcal{L} \big[\bm{a}(\hat{\theta}_k,{\Delta\bm{P}}),\bm{\mu}\big]=&\bm{a}(\hat{\theta}_k,{\Delta\bm{P}})^{\mathrm{H}}\bm{\bar{Q}}(\hat{\theta}_k ,\bm{\tilde{P}})\bm{a}(\hat{\theta}_k,{\Delta\bm{P}})\nonumber\\
		&+\bm{\mu}^{\mathrm{H}}\big\{\bm{W}^{\mathrm{H}}\bm{a}(\hat{\theta}_k,{\Delta\bm{P}})-\bm{1}_{M_c\times 1}\big\},
	\end{align}
	where $\bm{\mu}$ denotes the Lagrange multiplier vector. By taking the partial derivative of the Lagrange function with respect to $\bm{a}(\hat{\theta}_k,{\Delta\bm{P}})$ and equating it to zero, we obtain
	\begin{align}
		\frac{\partial \mathcal{L} }{\partial \bm{a}(\hat{\theta}_k,{\Delta\bm{P}})}= 2\bm{\bar{Q}}(\hat{\theta}_k ,\bm{\tilde{P}})\bm{a}(\hat{\theta}_k,{\Delta\bm{P}})+\bm{W}\bm{\mu} = \bm{0}_{M \times 1}.
	\end{align}
	This equation yields the optimal estimate $\bm{\hat{{a}}}(\hat{\theta}_k,{\Delta\bm{P}})$ given as
	\begin{align}\label{ahat}
		\bm{\hat{{a}}}(\hat{\theta}_k ,{\Delta\bm{P}})=-\bm{\bar{Q}}^{-1}{(\hat{\theta}}_k ,\bm{\tilde{P}})\bm{W}\bm{\mu}/2,
	\end{align}
	where $\bm\mu=-2\big[\bm{W}^{\mathrm{H}}\bm{\bar{Q}}(\hat{\theta}_k ,\bm{\tilde{P}})^{-1}\bm{W}\big]^{-1}\bm{1}_{M_c\times 1}$ is obtained by substituting \eqref{ahat} back into \eqref{condition}.
	
	Based on the phase information of the estimation, we can deduce the relationship for the $k$-th DOA of the $m$-th MA as
	\begin{align}\label{only}
		[\bm{\hat{\varTheta}} ]_{:,k}^{\mathrm{T}}[{\Delta \bm{\hat{P}}}]_{m,:}^{\mathrm{T}}=\kappa^{-1}\measuredangle [\bm{\hat{{a}}}(\hat{\theta}_k,{\Delta\bm{P}})]_m,
	\end{align}
	where $[\bm{\hat{\varTheta}} ]_{:,k}=[\cos \hat{\theta}_k ,\sin \hat{\theta}_k ]^{\mathrm{T}}$. It is worth noting that the APE in practical systems are intrinsically smaller than the operating wavelength. This physical reality naturally confines the resulting phase shift $\measuredangle [\bm{\hat{{a}}}(\hat{\theta}_k,{\Delta\bm{P}})]_m$ within the principal branch $(-\pi, \pi]$. This means the issue of phase wrapping ambiguity issue does not exist in practice. By concatenating the equations for all $K$ DOAs defined by $\bm{\hat{\theta}} =[\hat{\theta}_1 ,\hat{\theta}_2 ,\ldots,\hat{\theta}_K ]^{\mathrm{T}}$ and applying the least squares method, we arrive at
	\begin{align}\label{all}
		[{\Delta \bm{\hat{P}}}]_{m,:}^{\mathrm{T}}=\kappa^{-1} (\bm{\hat{\varTheta}}^{\mathrm{T}})^{\dagger} \measuredangle \bm{\hat{{b}}}(\bm{\hat{\theta}},{\Delta\bm{P}}_{m,:}),
	\end{align}
	where $\bm{\hat{{b}}}(\bm{\hat{\theta}},{\Delta\bm{P}}_{m,:})=\big[[\bm{\hat{{a}}}({\hat{\theta}}_1,{\Delta\bm{P}})]_m,\ldots,[\bm{\hat{{a}}}({\hat{\theta}}_K,{\Delta\bm{P}})]_m\big]^{\mathrm{T}}$. It is crucial to emphasize that estimating the array position errors along both the $x$ and $y$ axes requires the condition $K \geq 2$. For the special case of $K=1$, the formulation simplifies to estimating merely the gain/phase errors which has been thoroughly investigated in \cite{DOAGPEChengzhi2025}. Subsequently, we update the actual positions of MAs via $\bm{\hat{ P}}={\Delta \bm{\hat{P}}}+\bm{\tilde{P}}$.
	
	The estimated DOA $\bm{\hat{\theta}}^{(l)}$ at the $l$-th iteration directly corresponds to the $K$ principal maxima of $G$. The convergence criterion for the proposed algorithm is defined as $\|\bm{\hat{\theta}}^{(l)}-\bm{\hat{\theta}}^{(l-1)}\|_F^2\leq \delta$. Specific steps of the proposed AO-based self-calibration algorithm are summarized in Algorithm \ref{Algorithm 1}.
	\begin{algorithm}[t]
		\renewcommand{\algorithmicrequire}{\textbf{Input }}  
		\renewcommand{\algorithmicensure}{\textbf{Output }}  
		\caption{Proposed AO-Based Self-Calibration Algorithm}
		\label{Algorithm 1}
		\begin{algorithmic}[1]
			\Require 
			$\bm{Z}$, $T$, $K$, $H$, $M$, and $M_c$.
			\State Calculate the noise subspcace $\bm{E}_N$ via \eqref{RY} and \eqref{EVD};
			\State Construct the constraint condition via \eqref{condition};
			\While {$l\leq L$ or $\|\bm{\hat{\theta}}^{(l)}-\bm{\hat{\theta}}^{(l-1)}\|_F^2\geq \delta$}
			
			\Statex {\textbf{Stage 1: DOA Estimation}}
			\State Calculate the estimated DOA $\bm{\theta}^{(l)}$ via \eqref{Initial}.
			\Statex {\textbf{Stage 2: APE Estimation}}
			\For {each $k\in [1,\ldots,K]$ }
			
			\State Estimate $\bm{\hat{{a}}}(\hat{\theta}_k,{\Delta\bm{P}})$ via \eqref{ahat};
			\State Compute the function via  \eqref{only};
			\EndFor
			\State Estimate the APE of MA $[{\Delta \bm{\hat{P}}}^{(l)}]_{m,:}$ via \eqref{all};
			\State Update the estimated actual positions of MA $\bm{\hat{{P}}}^{(l-1)}$ via $\bm{\hat{ P}}^{(l)}={\Delta \bm{\hat{P}}}^{(l)}+\bm{\tilde{P}}$;
			\State $l\leftarrow l+1$;
			\EndWhile
			\Ensure Return estimated parameters $\bm{\hat{\theta}}$ and $\bm{\hat{ P}}$.
		\end{algorithmic}
	\end{algorithm}

	\vspace{-0.5em}
	\subsection{Computational Complexity Analysis}
	In this subsection, we evaluate the computational complexity of the proposed algorithm. The initialization stage, which encompasses the covariance matrix computation, eigenvalue decomposition, and the initial MUSIC spectral search, requires $\mathcal{O}(M^2 T + M^3 + GM(M-K))$ operations, where $G$ denotes search grids, respectively. Let $L_c$ denote the number of iterations required to reach convergence. In each iteration, the processes of updating the steering vectors, calculating the Lagrange multipliers, and refining the DOA estimation collectively demand $\mathcal{O}(KM^3 + KM_c^3 + KMM_c + GM(M-K))$ operations. Consequently, the overall computational complexity of the proposed algorithm evaluates to $\mathcal{O}\big(L_c(KM^3 + KM_c^3 + KMM_c + GM(M-K))\big)$. Consequently, the total computational complexity is formulated as
	\begin{align}\label{complexity}
		\mathcal{O}\big(&M^2 T + M^3 + GM(M-K) + \notag\\
		&L_c \big[KM^3 + KM_c^3 + KMM_c + GM(M-K)\big]\big),
	\end{align}
	which can be approximated as $\mathcal{O}\big(M^2 T + L_c(KM^3 + GM^2)\big)$ given $M_c \leq M$ and $K < M$. By avoiding prohibitive multi-dimensional spectral searches, the proposed iterative method maintains an acceptable computational overhead for practical MA systems.
	
	\vspace{-0.5em}
	\section{Numerical Simulations}
	In this section, we evaluate the performance of the proposed algorithm through numerical simulations. The movable region is set to $H = 12\lambda$, the number of MAs is set to $M = 12$, the optimized MA positions are obtained in \cite{ye2026movableantennaenhancedwireless}, the number of the calibrated MAs is $M_c=7$, the number of snapshots is $T=100$, and the number of the DOA is $K = 4$, which are randomly drawn from a uniform distribution $\mathcal{U}(\vartheta_{\min}, \vartheta_{\max})$, where we set $\vartheta_{\min} = \frac{\pi}{6}$ and $\vartheta_{\max} = \frac{5\pi}{6}$. 
	The estimation results are calculated by taking the average over 95\% of the realizations for mitigating the impact of outliers \cite{ISACRZzhang2024}. The APE is generated by 
	$\Delta x_m = \sigma_x \zeta_x$, and $\Delta y_m =\sigma_y \zeta_y$. $\zeta_x$ and $\zeta_y$ are randomly drawn from a uniform distribution $\mathcal{U}(-0.5, 0.5)$, and $\sigma_x=\sigma_y=0.5\lambda$ are the standard deviations of $\zeta_x$ and $\zeta_y$. The convergence tolerance is set as $\delta = 10^{-3}$. For convenient comparison, we plot the curve of conventional MUSIC algorithm utilizing all the MAs and the calibrated MAs, which are marked as \textbf{Baseline 1} and \textbf{Baseline 2}, respectively. To fully demonstrate the superiority of the proposed algorithm, we investigate the performance in scenarios with only the $x$-axis error, only the $y$-axis error, and simultaneous $x$- and $y$-axis errors, which are marked as \textbf{Baseline 3}, \textbf{Baseline 4}, and \textbf{Baseline 5}, respectively.


\begin{figure}[t]
	\vspace{-0.5em}
	\centering
	\begin{minipage}[t]{0.45\linewidth}
		\centering
		\includegraphics[width=\linewidth]{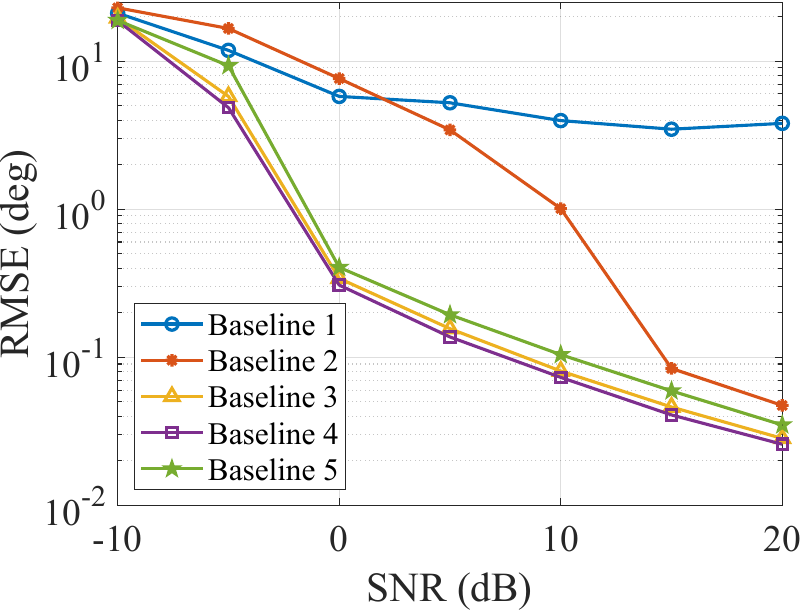}
		\caption{RMSE versus SNR.}
		\label{theta_SNR}
	\end{minipage}
	\hfill
	\begin{minipage}[t]{0.45\linewidth}
		\centering
		\includegraphics[width=\linewidth]{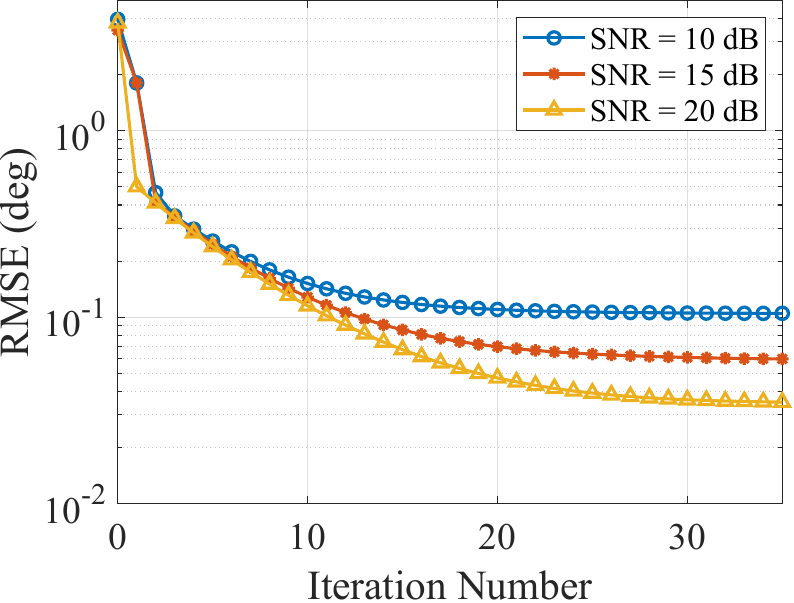}
		\caption{RMSE versus iteration numbers.}
		\label{fig_iter}
	\end{minipage}
	\vspace{-0.5em}
\end{figure}

\begin{figure}[t]
	\centering
	\begin{minipage}[t]{0.45\linewidth}
		\centering
		\includegraphics[width=\linewidth]{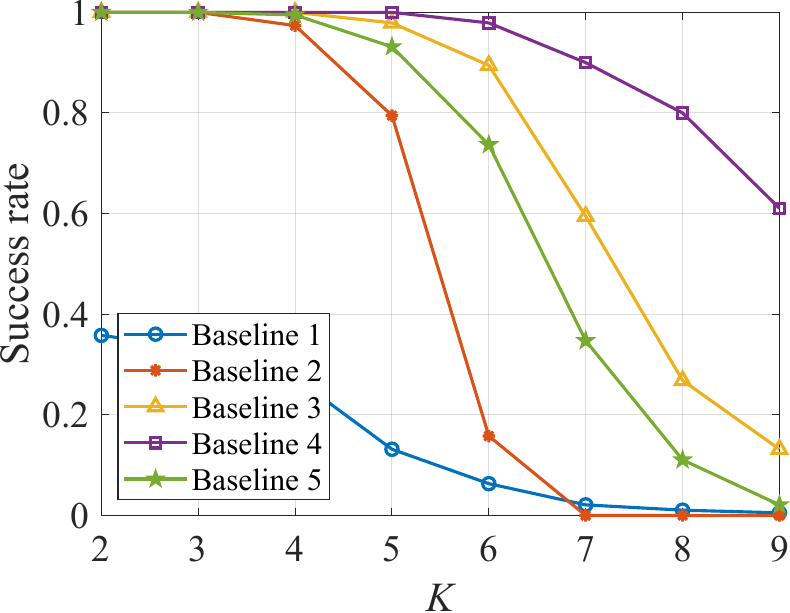}
		\caption{Success rate versus $K$.}
		\label{theta_number}
	\end{minipage}
	\hfill
	\begin{minipage}[t]{0.45\linewidth}
		\centering
		\includegraphics[width=\linewidth]{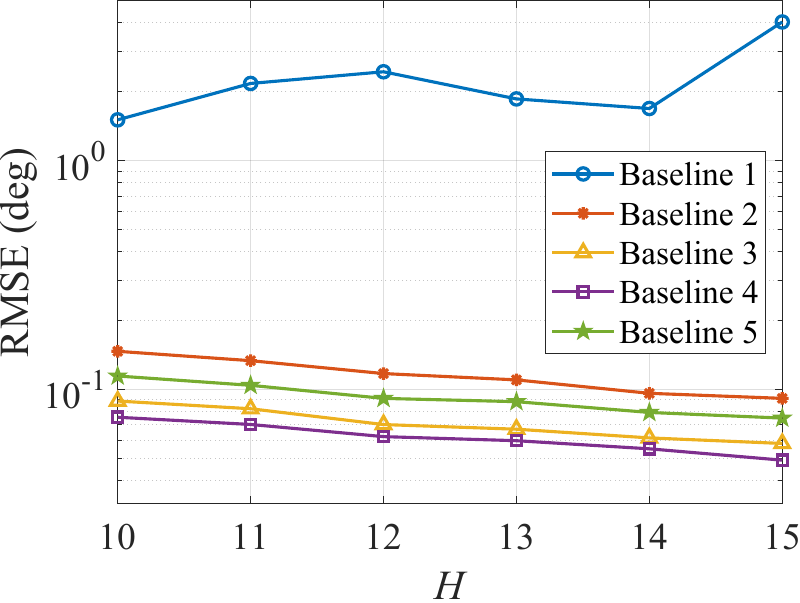}
		\caption{Success rate versus $H$.}
		\label{RMSE_H}
	\end{minipage}
	\vspace{-0.5em}
\end{figure}

\begin{figure}[t]
	\vspace{-0.5em}
	\centering
	\begin{minipage}[t]{0.45\linewidth}
		\centering
		\includegraphics[width=\linewidth]{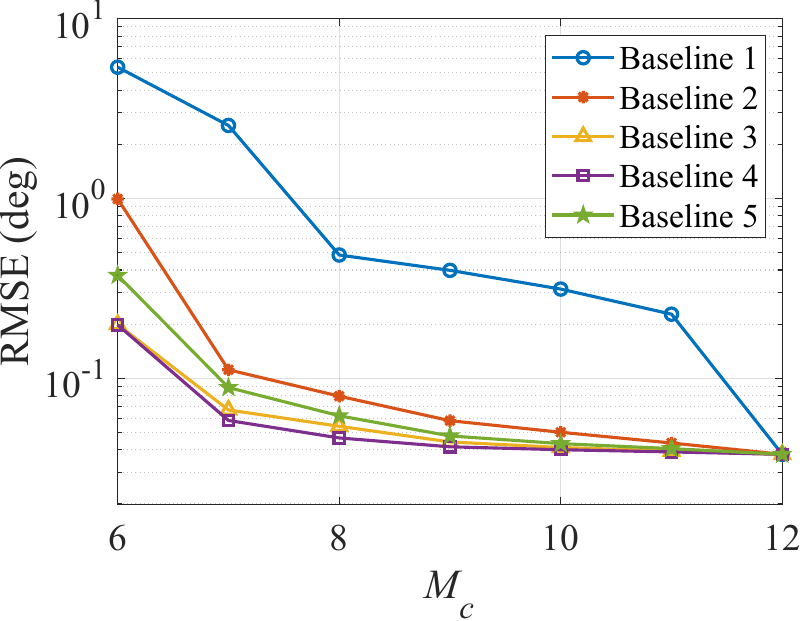}
		\caption{Success rate versus $M_c$.}
		\label{RMSE_Mc}
	\end{minipage}
	\hfill
	\begin{minipage}[t]{0.45\linewidth}
		\centering
		\includegraphics[width=\linewidth]{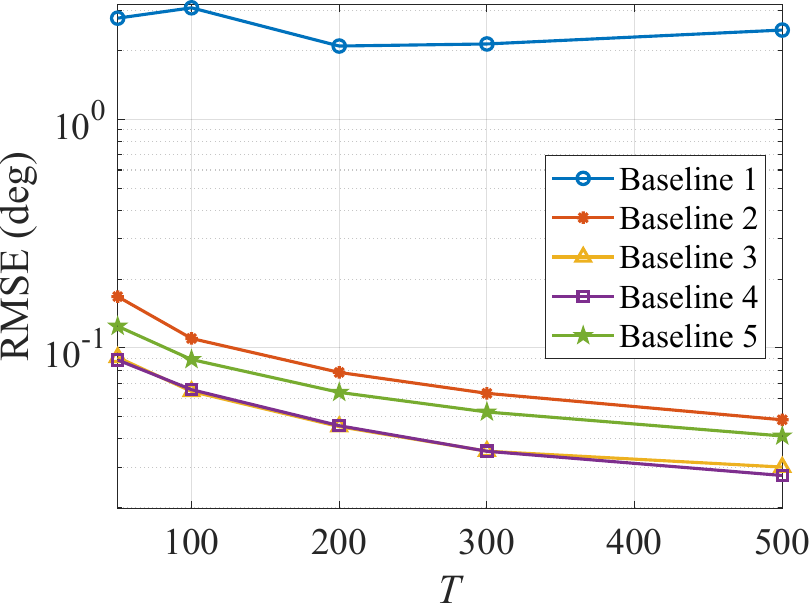}
		\caption{RMSE versus $T$.}
		\label{RMSE_T}
	\end{minipage}
	\vspace{-0.5em}
\end{figure}
	
	\Cref{theta_SNR} illustrates the variation in the RMSE of DOA estimation with respect to the SNR. In the presence of APE, the performance of the conventional MUSIC algorithm utilizing all MAs fails to achieve valid DOA estimation. In contrast, RMSE of the proposed iterative algorithms decreases significantly with SNR increases. Specifically, the performance of the proposed algorithm with errors confined to the $y$-axis outperforms the case with APE only in the $x$-axis, which in turn surpasses the configuration with errors in both the $x$ and $y$ axes. Although the RMSE of the conventional MUSIC algorithm using calibrated MAs also decreases with the increase of SNR, it consistently remains higher than that of the proposed algorithm. Furthermore, the MUSIC algorithm using only calibrated MAs demonstrates particularly poor performance in the low SNR regime.
	
	\Cref{fig_iter} demonstrates the convergence performance of the proposed Algorithm. We can observe that the DOA estimation performance of the proposed algorithm gradually improves as the number of iterations increases, which indicates that the proposed algorithm can gradually converge to a stationary solution as the number of iterations increases. Specifically, it takes approximately 30 iterations for convergence.
	
	\Cref{theta_number} illustrates the success rate of DOA estimation versus the number of sources. A target is considered successfully estimated
	if the DOA estimation error satisfies $|\theta_k-\hat{\theta}_{n,k}|\leq0.5^{\circ}$. In the presence of APE, conventional MUSIC algorithm fails to estimate the DOA. Owing to the limitation on the number of calibrated antennas, the DOA cannot be estimated using only calibrated antennas when $K\geq M_c$, and the success rate of DOA estimation drops sharply when $K\geq2$. In addition, when $2 \le K\leq6$, the proposed algorithm can maintain a high DOA estimation success rate under the cases with APE.
	
	\Cref{RMSE_H} decipts the RMSE of the evaluated algorithms in the presence of varying movable region sizes $H$ where the estimation error is shown as a function of the movement range. We can observe that the performance of the proposed algorithm improves as the movable region increases. Specifically the conventional MUSIC algorithm utilizing all MAs fails to achieve reliable estimation and fluctuates severely across the evaluated range. The proposed algorithm maintains an accurate estimation under all of the imperfect cases. 
	
	\Cref{RMSE_Mc} demonstrates the RMSE with varying numbers of calibrated antennas. We can observe that the performance of all evaluated methods improves as $M_c$ increases. Specifically, the MUSIC algorithm utilizing the all MAs suffers from severe performance degradation when the uncalibrated proportion is large and it only achieves acceptable accuracy when all antennas are perfectly calibrated. The proposed self-calibration approach effectively achieve robust DOA estimation under all imperfect cases.
	
 \Cref{RMSE_T} shows the RMSE performance varying numbers of snapshots $T$. We can observe that the DOA estimation performance of the proposed algorithm improves as the number of snapshots increases. Although the algorithm using only calibrated antennas shows improved accuracy it consistently presents a higher error bound. The proposed algorithm consistently maintains a robust accurate estimation under the cases with APE even when the number of snapshots is limited.
	
	\vspace{-0.5em}

\section{Conclusion}
In this letter, we addressed the DOA estimation problem for MA systems in the presence of unknown APE. By exploiting the orthogonality between the steering vector and the noise subspace, we formulated the joint estimation problem of DOA and APE as an optimization problem and developed an AO-based self-calibration algorithm. In each iteration, the DOA is initially updated using the MUSIC algorithm with the current APE estimate, while the APE is updated in closed-form derived via the Lagrange multiplier method. Simulation results demonstrated that the proposed algorithm significantly outperforms conventional approaches, achieving robust and accurate DOA estimation.

	
	
	%
	\vspace{-0.5em}
	\raggedbottom 
	\bibliographystyle{IEEEtran}
	\bibliography{mybib_Abbreviation_1}

\end{document}